\documentclass{llncs}
\usepackage{amssymb,amsmath,mathtools}
\usepackage{xcolor}
\usepackage{setspace}
\usepackage{fullpage}

\newtheorem {myclaim}{Claim}
\newcommand{\layer}{{\mathcal L}}

\begin{document}
\title{Fast deterministic algorithms for computing all eccentricities in (hyperbolic) Helly graphs}

\author{Feodor F. Dragan\inst{1} \and
Guillaume Ducoffe\inst{2} \and
Heather M. Guarnera\inst{3}}

\institute{Computer Science Department, Kent State University, Kent, USA \\
\email{dragan@cs.kent.edu}  
\and 
National Institute for Research and
Development in Informatics and University
of Bucharest, Bucure\c{s}ti, Rom\^{a}nia\\
\email{guillaume.ducoffe@ici.ro} 
\and 
Department of Mathematical and Computational Sciences, The College of Wooster, Wooster, USA \\
\email{hguarnera@wooster.edu}
}

\maketitle

\begin{abstract}
A graph is Helly if every family of pairwise intersecting balls has a nonempty common intersection. The class of Helly graphs is the discrete analogue of the class of hyperconvex metric spaces. It is also known that every graph isometrically embeds into a Helly graph, making the latter an important class of graphs in Metric Graph Theory. We study diameter, radius and all eccentricity computations within the Helly graphs. Under plausible complexity assumptions, neither the diameter nor the radius can be computed in truly subquadratic time on general graphs. In contrast to these negative results, it was recently shown that the radius and the diameter of an $n$-vertex $m$-edge Helly graph $G$ can be computed with high probability in $\tilde{\mathcal O}(m\sqrt{n})$ time ({\it i.e.}, subquadratic in $n+m$). 
In this paper, we improve that result by presenting a deterministic ${\mathcal O}(m\sqrt{n})$ time algorithm which computes not only the radius and the diameter but also all vertex eccentricities in a Helly graph. 
Furthermore, we give a parameterized linear-time algorithm for this problem on Helly graphs, with the parameter being the Gromov hyperbolicity $\delta$. More specifically, we show that the radius and a central vertex of an $m$-edge $\delta$-hyperbolic Helly graph $G$ can be computed in $\mathcal O(\delta m)$ time and that all vertex eccentricities in $G$ can be computed in $\mathcal O(\delta^2 m)$ time. To show this more general result, we heavily use our new structural properties obtained for Helly graphs. 
\end{abstract}


\section{Introduction} 
Given an undirected unweighted graph $G=(V,E)$, the distance $d_G(u,v)$ between two vertices $u$ and $v$ is the minimum number of edges on any path connecting $u$ and $v$ in $G$. The eccentricity $e_G(u)$ of a vertex $u$ is the maximum distance from $u$ to any other vertex. The radius and the diameter of $G$, denoted by $rad(G)$ and $diam(G)$, are the smallest and the largest eccentricities of vertices in $G$, respectively. A vertex with eccentricity equal to $rad(G)$ is called a central vertex of $G$.  We are interested in the fundamental problems of finding a central vertex and of computing the diameter and the radius of a graph.  The problem of finding a central vertex of a graph is one of the most famous facility location problems in Operation Research and in Location Science. 
The diameter and radius of a graph play an important role in the design and analysis of networks in a variety of networking environments like social networks, communication networks, electric power grids, and transportation networks.
A naive algorithm which runs breadth-first-search from each vertex to compute its eccentricity and then (in order to compute the radius, the diameter and a central vertex) picks the  smallest and the largest eccentricities and a vertex with smallest eccentricity has running time ${\mathcal O}(nm)$ on an $n$-vertex $m$-edge graph. Interestingly, this naive algorithm is conditionally optimal for general graphs as well as for some restricted families of graphs \cite{AVW16,BCH16,JGAA-Dragan,RoV13} since, under plausible complexity assumptions, neither the diameter nor the radius can be computed in truly subquadratic time on those graphs. Already for split graphs (a subclass of chordal graphs), computing the diameter is roughly equivalent to {\sc Disjoint Sets}, {\it a.k.a.}, the monochromatic {\sc Orthogonal Vector}   problem~\cite{ChD92}. Under the Strong Exponential-Time Hypothesis (SETH), we cannot solve {\sc Disjoint Sets} in truly subquadratic time~\cite{Wil05}, and so neither we can compute the diameter of split graphs in truly subquadratic time~\cite{BCH16}.

In a quest to break this quadratic barrier (in the size $n+m$ of the input), there has been a long line of work presenting more efficient algorithms for computing the diameter and/or the radius on some special graph classes, by exploiting their geometric and tree-like representations and/or some forbidden pattern ({\it e.g.}, excluding a minor, or a family of induced subgraphs). 
For example, although the diameter of a split graph can unlikely be computed in subquadratic time, there is an elegant linear-time algorithm for computing the radius and a central vertex of a chordal graph \cite{ChepoiD94-esa}. Efficient algorithms  for computing the diameter and/or the radius or finding a central vertex are also known for interval graphs~\cite{DraganNB96-wg,Ola90}, AT-free graphs~\cite{DBLP:journals/corr/abs-2010-15814},   directed path graphs~\cite{CDHP01}, distance-hereditary graphs~\cite{CDP18,Dragan94-swat,DraganG20-tcs,DraganN00-dam}, strongly chordal  graphs~\cite{DraganPhD}, dually chordal graphs~\cite{BrandstadtCD98-dam,Dragan1993-CSJofM},  chordal bipartite graphs~ \cite{DBLP:journals/corr/abs-2101-03574}, outerplanar graphs~\cite{FaP80}, planar graphs~\cite{Cab18,GKHM+18}, graphs with bounded clique-width~\cite{CDP18,DBLP:journals/corr/abs-2011-08448}, 
graphs with bounded tree-width~\cite{AVW16,BHM18,DHV19+a} and, more generally,  $H$-minor free graphs and graphs of bounded (distance) VC-dimension~\cite{DHV19+a}. See also~\cite{CDV02,JGAA-Dragan,Duc19,DBLP:journals/corr/abs-2010-15803,DHV19+b,Epp00} for other examples.

We here study the {\em Helly graphs} as a broad generalization of dually chordal graphs which in turn contain all interval graphs, directed path graphs and strongly chordal graphs. Recall that a graph is Helly if every family of pairwise intersecting balls has a non-empty common intersection.
This latter property on the balls will be simply referred to as the Helly property in what follows.
Helly graphs have unbounded tree-width and unbounded clique-width, they do not exclude any fixed minor and they cannot be characterized via some forbidden structures. 
They are sometimes called absolute retracts or disk-Helly graphs by opposition to other Helly-type properties on graphs~\cite{DPS09}. The Helly graphs are well studied in Metric Graph Theory. {\it E.g.}, see the survey~\cite{BaC08} and the papers cited therein. This is partly because every graph is an isometric subgraph of some Helly graph, thereby making of the latter the discrete equivalent of hyperconvex metric spaces~\cite{Dre84,Isb64}. 
A minimal by inclusion Helly graph $H$  which contains a given graph $G$ as an isometric subgraph is unique and called the injective hull~\cite{Isb64} or the tight span~\cite{Dre84} of $G$. 
Polynomial-time recognition algorithms for the Helly graphs were presented in~\cite{BaP89,DraganPhD,LiS07}. Several structural properties of these graphs were also identified  in~\cite{BaP89,BaP89b,BaP91,CLP00,DraganPhD,LocGlob,Dra93,DrB96,DrGu19,Pol01,Pol03}. The dually chordal graphs are exactly the Helly graphs in which the intersection graph of balls is chordal, and they were studied independently from the general Helly graphs~\cite{BrandstadtCD98-dam,BDCV98,DPC92,Dragan1993-CSJofM,DrB96,SzwarcfiterB94-siamdm}. As we already mentioned it~\cite{BrandstadtCD98-dam,DPC92,Dragan1993-CSJofM}, the diameter, the radius and a central vertex of a dually chordal graph can be found in linear time, that is optimal. However, it was open until recently whether there are truly subquadratic-time algorithms for these problems on general Helly graphs.  First such  algorithms were recently presented in~\cite{DrDu2019-HellyStory} 
for computing both the radius and the diameter and in~ \cite{Du2020-CentMed} for finding a central vertex. Those algorithms are randomized and run,  with high probability,  in $\tilde{\mathcal O}(m\sqrt{n})$ time on a given $n$-vertex $m$-edge Helly graph ({\it i.e.}, subquadratic in $n+m$). They make use of the Helly property and of the unimodality of the eccentricity function in Helly graphs~\cite{LocGlob}: every vertex of locally minimum eccentricity is a central vertex. In~\cite{DrDu2019-HellyStory}, a linear-time algorithm for computing all eccentricities in $C_4$-free Helly graphs was also presented. The $C_4$-free Helly graphs are exactly the Helly graphs whose balls are convex. They properly include strongly chordal graphs as well as bridged Helly graphs and hereditary Helly graphs~\cite{DrDu2019-HellyStory}. 

\paragraph*{\bf Our Contribution.}
We improve those results from ~\cite{DrDu2019-HellyStory} and 
\cite{Du2020-CentMed} by presenting a deterministic ${\mathcal O}(m\sqrt{n})$ time algorithm which computes not only the radius and the diameter but also all vertex eccentricities in an $n$-vertex $m$-edge Helly graph. This is our first main result in the paper. Being able to efficiently compute all vertex eccentricities is of great importance. For example, in the analysis of social networks
(e.g., citation networks or recommendation networks), biological systems (e.g., protein interaction networks),
computer networks (e.g., the Internet or peer-to-peer networks), transportation networks (e.g., public transportation
or road networks), etc., the eccentricity $e_G(v)$ of a vertex $v$ is used to measure its importance in the network: the {\em eccentricity centrality index} of $v$ \cite{Brandes} is defined as $\frac{1}{e_G(v)}$. 

Our second main result is a parameterized {\em linear-time} algorithm for computing all vertex eccentricities in Helly graphs, with the parameter being the Gromov hyperbolicity $\delta$, as defined by the following four point condition. The hyperbolicity of a graph $G$ \cite{Gromov1987} is the smallest half-integer $\delta\ge 0$  such that,  for any four vertices $u,v,w,x$, the two largest of the three distance sums $d(u,v)+d(w,x)$, $d(u,w)+d(v,x)$, $d(u,x)+d(v,w)$ differ by at most $2\delta$. In this case we say that $G$ is $\delta$-hyperbolic. As the tree-width of a graph measures its combinatorial tree-likeness, so does the hyperbolicity of a graph measure its metric tree-likeness. In other words, the smaller the hyperbolicity $\delta$ of $G$ is, the closer $G$ is to a tree metrically. The hyperbolicity of an $n$-vertex graph can be computed in polynomial-time (e.g., in ${\mathcal O}(n^{3.69})$ time~\cite{DBLP:journals/ipl/FournierIV15}), however it is unlikely that it can be done in subquadratic  time~\cite{BCH16,CoudertD14-siamdm,DBLP:journals/ipl/FournierIV15}. A 2-approximation of hyperbolicity can be computed in ${\mathcal O}(n^{2.69})$ time~\cite{DBLP:journals/ipl/FournierIV15} and an 8-approximation can be computed in ${\mathcal O}(n^2)$ time~\cite{DBLP:conf/compgeom/ChalopinCDDMV18} (assuming that the input
is the distance matrix of the graph). Graph hyperbolicity has attracted attention recently due to the empirical evidence that it takes small values in many real-world networks, such as biological networks, social networks, Internet application networks, and collaboration networks, to name a few (see, e.g., \cite{Abu-AtaD16-networks,AdcockSM13,Bo++,KeSN16,NaSa,ShavittT08}).
Furthermore, many special graph classes (e.g., interval graphs, chordal graphs, dually chordal graphs, AT-free graphs, weakly chordal graphs and many others) have constant hyperbolicity \cite{Abu-AtaD16-networks,BandeltC03-siamdm,uea21813,ChepoiDEHV08-SODA2008,DrGu19,DraganM19-dmtcs,KoolenM02-ejc,WuZ11-combinatorics}. In fact, the dually chordal graphs and the $C_4$-free Helly graphs are known to be proper subclasses of the $1$-hyperbolic Helly graphs (this follows from results in \cite{BDCV98,DrGu19}). Notice also that any graph is $\delta$-hyperbolic for some $\delta\le diam(G)/2$.

We show that the radius and a central vertex of an $m$-edge Helly graph $G$ with hyperbolicity $\delta$ can be computed in $\mathcal O(\delta m)$ time and that all vertex eccentricities in $G$ can be computed in ${\mathcal O}(\delta^2 m \log\delta)$ time, even if $\delta$ is not known to us. If either $\delta$ or a constant approximation of it is known, then the running time of our algorithm can be lowered to $\mathcal O(\delta^2 m)$. Thus, for Helly graphs with constant hyperbolicity, all vertex eccentricities can be computed in linear time.  
As a byproduct, we get a linear time algorithm for computing all eccentricities in $C_4$-free Helly graphs as well as in dually chordal graphs, generalizing known results from~\cite{BrandstadtCD98-dam,Dragan1993-CSJofM,DrDu2019-HellyStory}. Previously, for dually chordal graphs, it was only known  that a central vertex can be found in linear time~\cite{BrandstadtCD98-dam,Dragan1993-CSJofM}.
Notice that the diameter problem can unlikely be solved in truly subquadratic time in general 1-hyperbolic graphs and that the radius problem can unlikely be solved in truly subquadratic time in general 2-hyperbolic graphs~\cite{JGAA-Dragan}.
For general $\delta$-hyperbolic graphs, there are only additive $\mathcal O(\delta)$-approximations of the diameter and the radius, that can be computed in linear time~\cite{ChepoiDEHV08-SODA2008,EccentricityTerrain,DrHaVi2018-Revisit}. 

To show our more general results, additionally to the unimodality of the eccentricity function in Helly graphs, we rely on new structural properties obtained for this class. 
It turns out that the hyperbolicity of a Helly graph $G$ is governed by the size of a largest isometric rectilinear grid in $G$. As a consequence, the hyperbolicity of an $n$-vertex Helly graph is at most $\sqrt{n} + 1$ and the diameter of the center $C(G)$ of $G$  is at most $2 \sqrt{n} + 3$. These properties, along with others, play a crucial role in efficient computations of all eccentricities in Helly graphs. 
We also give new characterizations of the class of Helly graphs. Among others, we show that the  Helly  property  for  balls  of  equal  radii  implies  the  Helly  property  for  balls  with variable radii.  It would be interesting to know whether a similar result holds for all (discrete) metric spaces.  We are not aware of such a general result.  


\paragraph*{\bf Notations.}
Recall that $d_G(u,v)$ denotes the distance between vertices $u$ and $v$ in $G=(V,E)$.  Let $n=|V|$ be the number of vertices and  $m=|E|$ be the number of edges in $G$.  
The ball of radius $r$ and center $v$ is defined as $\{ u \in V : d_G(u,v) \leq r \}$, and denoted by $N_G^r[v]$. Sometimes, $N_G^r[v]$ is called the $r$-neighborhood of $v$. In particular, $N_G[v] := N_G^1[v]$ and $N_G(v) := N_G[v] \setminus \{v\}$ denote the closed and open neighbourhoods of a vertex $v$, respectively.
More generally, for any vertex-subset $S$ and a vertex $u$, we define $d_G(u,S) := \min_{v \in S} d_G(u,v), \ N_G^r[S] := \bigcup_{v \in S}N_G^r[v], \ N_G[S] := N_G^1[S] \ \text{and} \ N_G(S) := N_G[S] \setminus S$.
The metric projection of a vertex $u$ on $S$, denoted by $Pr_G(u,S)$, is defined as $\{ v \in S : d_G(u,v) = d_G(u,S) \}$. The metric interval $I_G(u,v)$ between $u$ and $v$ is $\{ w \in V : d_G(u,w) + d_G(w,v) = d_G(u,v)\}$.
For any $k \leq d_G(u,v)$, we can also define the slice $L(u,k,v) := \{ w \in I_G(u,v) : d_G(u,w) = k\}$.
Recall that the eccentricity of a vertex $u$ is defined as $\max_{v \in V} d_G(u,v)$ and denoted by $e_G(u)$.
Note that we will omit the subscript if the graph $G$ is clear from the context. 
The radius and the diameter of a graph $G$ are denoted by $rad(G)$ and $diam(G)$, respectively. A vertex $c$ is called central in $G$ if $e_G(c)=rad(G)$. 
The set of all central vertices of $G$ is denoted by $C(G) := \{ v \in V : e_G(v) = rad(G) \}$ and called the center of $G$. The eccentricity  function $e_G(v)$ of a graph $G$ is said to be {\it  unimodal}, if for every non-central vertex $v$ of $G$ there is a neighbor $u\in N_G(v)$ such that $e_G(u)<e_G(v)$ (that is, every local minimum of the eccentricity function is a global minimum). 
Recall also that a vertex set $S\subseteq V$ is called convex in $G$ if, for every vertices $x,y\in S$,  all shortest paths connecting them are contained in $S$ (i.e., $I_G(x,y)\subseteq S$). For $\beta\ge 0$, we say that $S$ is   $\beta$-\emph{pseudoconvex}~\cite{EccentricityTerrain} if, for every vertices $x,y \in S$, any vertex $z \in I_G(x,y) \setminus S$  satisfies $\min\{d_G(z,x), d_G(z,y)\} \leq \beta$. A subgraph $H$ of $G$ is called isometric (or distance-preserving) if, for every vertices $x,y$ of $H$,  $d_G(x,y)=d_{H}(x,y)$.  


\section{Characterizations of Helly graphs and hyperbolicity in Helly graphs}\label{sec:charc}
Here we demonstrate that for  Helly graphs, having a constant hyperbolicity is equivalent to the following properties: having $\beta$-pseudoconvexity of balls with a constant $\beta$, or having the diameter of the center bounded by a constant for all subsets of vertices, or not having a large $(\gamma \times \gamma)$ rectilinear grid as an isometric subgraph. These results generalize some known results from \cite{chalopin2020helly,ChepoiDEHV08-SODA2008,DrGu19,EccentricityTerrain}. 
 
\medskip

First we give new characterizations of Helly graphs through a formula for the eccentricity function and relations between diameter and radius for all subsets of vertices. For this we need to generalize our basic notations. Define for any set $M\subseteq V$ and any vertex $v\in V$ the eccentricity of $v$ in $G$ with respect to $M$ as follows: $$e_M(v)=\max_{u\in M}d_G(u,v).$$ Let $diam_M(G)=\max_{v\in M}e_M(v)$, $rad_M(G)=\min_{v\in V}e_M(v)$,  $C_M(G)=\{v\in V: e_M(v)=rad_M(G)\}$. When $M=V$, these agree with earlier definitions.

\begin{theorem} \label{th:charcter} For a graph $G$ the following statements are equivalent: 
\begin{itemize}
\item[(1)]  $G$ is Helly;

\item[(2)] the eccentricity function $e_M(\cdot)$ is unimodal for every set $M\subseteq V$; 

\item[(3)] $e_M(v)=d_G(v,C_M(G))+rad_M(G)$ holds for every set $M\subseteq V$ and every vertex $v\in V$;

\item[(4)] $2rad_M(G)-1\le diam_M(G) \le 2rad_M(G)$ holds for every set $M\subseteq V$;

\item[(5)] $rad_M(G)=\lfloor \frac{diam_M(G)+1}{2}\rfloor$ holds for every set $M\subseteq V$.

\end{itemize}
\end{theorem}

\begin{proof} The equivalence (1)$\Leftrightarrow$(2) was proven in \cite{DraganPhD,LocGlob}. It was also shown in \cite{DraganPhD} that (2) and (3) are equivalent for $M=V$. Here, we complete the proof of all equivalencies. We will show (2)$\Leftrightarrow$(3), when one considers any set $M\subseteq V$, and (1)$\Rightarrow$(4)$\Rightarrow$(5)$\Rightarrow$(1).

(2)$\Rightarrow$(3): 
Let $M$ be any subset of $V$. We will prove the formula in (3) by induction on $k=e_M(v)-rad_M(G)$. If $k=0$ then $e_M(v)=rad_M(G)$, i.e., $v\in C_M(G)$,  and the formula is trivially correct. Consider now a vertex $v$ with $e_M(v)>rad_M(G)$. By the triangle inequality, $e_M(v)\le d_G(v,C_M(G))+rad_M(G)$ always holds. As the eccentricity function $e_M(\cdot)$ is unimodal, there is a neighbor $u$ of $v$ with $e_M(v)>e_M(u)$. By induction hypothesis, $e_M(u)=d_G(u,C_M(G))+rad_M(G)$. Hence, by the triangle inequality, $e_M(v)\ge e_M(u)+1=d_G(u,C_M(G))+rad_M(G)+1\ge d_G(v,C_M(G))+rad_M(G)$. Combining both inequalities, we get $e_M(v)=d_G(v,C_M(G))+rad_M(G)$.

(3)$\Rightarrow$(2): 
Let $M$ be any subset of $V$ and $v$ be an arbitrary vertex of $G$ with $e_M(v)>rad_M(G)$. Let also $c_v$ be a vertex of $C_M(G)$ closest to $v$ and $u$ be a neighbor of $v$ on a shortest path from $v$ to $c_v$. We have $e_M(v)=d_G(v,c_v)+rad_M(G)= 1+d_G(u,c_v)+rad_M(G)\ge 1+d_G(u,C_M(G))+rad_M(G)=1+e_M(u)$. In particular, $e_M(v) > e_M(u)$.

(1)$\Rightarrow$(4): 
It is clear that  $diam_M(G) \le 2rad_M(G)$ holds for every graph $G$ and every set $M\subseteq V$ (by the triangle inequality). Assume now that for a subset $M$ of $V$, $diam_M(G)\le 2rad_M(G)-2$ holds. For every vertex $v$ in $M$, consider a ball centered at $v$ and with radius $rad_M(G)-1$. All these balls pairwise intersect as $diam_M(G)\le 2rad_M(G)-2$. By the Helly property, there must exist a vertex $c$ in $G$ such that $d_G(c,v)\leq rad_M(G)-1$ for every $v\in M$. The latter implies that $e_M(c)\leq rad_M(G)-1$, giving a contradiction.     

(4)$\Rightarrow$(5): 
It is straightforward.

(5)$\Rightarrow$(1): 
We know that a graph $G$ is Helly if and only if the family of unit balls of $G$ has the Helly property and for any three vertices $x,y,v$ with $d_G(x,y)\le 2$ and $d_G(x,v)=d_G(y,v)=k\ge 2$ there exists a common neighbor $u$ of $x$ and $y$ such that $d_G(v,u)=k-1$ (see \cite{BaP89,DraganPhD,LocGlob}).

Consider in $G$ a family $\mathfrak{F}$ of pairwise intersecting unit balls with centers at vertices $v_1,\dots,v_q$. Define $M=\{v_1,\dots,v_q\}$. As the balls pairwise intersect, $d_G(v_i,v_j)\leq 2$ for each $i,j\in \{1,\dots,q\}$. Hence, $diam_M(G)\leq 2$ and therefore $rad_M(G)=\lfloor \frac{diam_M(G)+1}{2}\rfloor\leq 1$. The latter implies the existence of a vertex $c$ in $G$ with $d_G(c,v_i)\leq 1$ for all $i\in \{1,\dots,q\}$. Necessarily, $c$ belongs to all unit balls from $\mathfrak{F}$. In other words, the family of unit balls of $G$ has the Helly property.

Let $x,y,v$ be any three vertices of $G$ with $d_G(x,y)\le 2$ and $d_G(x,v)=d_G(y,v)=k\ge 2$. We will show by induction on $k$ that there exists a common neighbor $u$ of $x$ and $y$ such that $d_G(v,u)=k-1$. If $k=2$ then the existence of $u$ follows from the Helly property for the family of unit balls of $G$.  Assume  now that $k>2$ and first consider the case when $k$ is even, say $k=2\ell$. Let $M=\{x,y,v\}$. As $k>2$, $diam_M(G)=k=2\ell$ and therefore $rad_M(G)=\lfloor \frac{diam_M(G)+1}{2}\rfloor=\ell$. Hence, there is a vertex $c$ in $G$ such that $d_G(c,v)=d_G(c,x)=d_G(c,y)=\ell$. Notice that $\ell \geq 2$. By induction, there must exist a common neighbor  $u$ of $x$ and $y$ such that $d_G(c,u)=\ell-1$. Necessarily, $d_G(u,v)=2\ell -1=k-1$. Assume now that $k$ is odd, say $k=2\ell+1$. If $k\ge 5$ then consider a neighbor $x'$ of $x$ on a shortest path from $x$ to $v$ and a neighbor $y'$ of $y$ on a shortest path from $y$ to $v$. Let $M=\{x',y',v\}$. As $d_G(x',v)=d_G(y',v)=k-1\ge 4$ and $d_G(x',y')\leq 4$, $diam_M(G)=k-1=2\ell$ and therefore $rad_M(G)=\lfloor \frac{diam_M(G)+1}{2}\rfloor=\ell$. Hence, there is a vertex $c$ in $G$ such that $d_G(c,v)=d_G(c,x')=d_G(c,y')=\ell$. Since $d_G(c,x)=d_G(c,y)=\ell +1$, by induction, there must exist a common neighbor  $u$ of $x$ and $y$ such that $d_G(c,u)=\ell$. Necessarily, $d_G(u,v)=2\ell=k-1$. 

It remains to consider the case when $k=3$. Let $M=I_G(x,v)\cup I_G(y,v)$. We have $3\le diam_M(G)\le 4$ and therefore $rad_M(G)=2$. The latter implies existence of a vertex $c$ in $G$ such that $d_G(c,w)\le 2$ for all $w\in M$. Consider a neighbor $v_x$ of $v$ on a shortest path from $v$ to $x$ and a neighbor $v_y$ of $v$ on a shortest path from $v$ to $y$. Since $d_G(c,v_x)\le 2$, $d_G(c,x)\le 2$ and $d_G(x,v_x)= 2$, the three unit balls centered at vertices $x,c,v_x$ pairwise intersect. By the Helly property for unit balls, there must exist a vertex $x'$ in $G$ which is adjacent to $x$ and $v_x$ and at distance at most 1 from $c$. By symmetry, there exists also a vertex $y'$ in $G$ which is adjacent to $y$ and $v_y$ and at distance at most 1 from $c$. Since $d_G(v,x')=d_G(v,y')=2$ and 
$d_G(x',y')\le d_G(c,x')+d_G(c,y')\le 2$, by induction, there is a vertex $v'$ in $G$ adjacent to all $x',y',v$. Applying again the induction hypothesis to $v',x,y$, we get a vertex $u$ in $G$ which is adjacent to all $v',x,y$ and hence at distance 2 from $v$. This concludes the proof. 
\qed\end{proof}

The equivalence between (1) and (5) can be rephrased as follows. 
\begin{corollary}\label{cor:Helly-uniform-to-all}
For every graph $G=(V,E)$, the family of all balls $\{N^r_G[v]: v\in V, r\in \mathbb{N}\}$ 
of $G$ has the Helly property if and only if the family of $k$-neighborhoods $\{N^k_G[v]: v\in V\}$ of $G$  has the Helly property for every natural number $k$.   
\end{corollary}

That is, the Helly property for balls of equal radii implies the Helly property for balls with variable radii. It would be interesting to know whether a similar result holds for all (discrete) metric spaces. We are not aware of such a general result and did not find its analog in the literature. 

\begin{proof} It is sufficient to prove that if  the family of $k$-neighborhoods $\{N^k_G[v]: v\in V\}$ of $G$ has the Helly property for every natural number $k$, then $G$ satisfies the condition (5) of Theorem \ref{th:charcter}. 

Consider an arbitrary set $M\subseteq V$. Denote $k:=\lfloor \frac{diam_M(G)+1}{2}\rfloor$. Since $d_G(x,y)\le diam_M(G)$ for every pair $x,y$ of vertices from $M$, the family of $k$-neighborhoods $\{N^k_G[v]: v\in M\}$ of $G$ consists of pairwise intersecting sets. By the Helly property, there is a vertex $c\in V$ which belongs to all those $k$-neighborhoods. Necessarily, $d_G(c,v)\leq k$ holds for every $v\in M$. Hence, $rad_M(G)\le k=\lfloor \frac{diam_M(G)+1}{2}\rfloor$. As $diam_M(G)\leq 2rad_M(G)$, we get $rad_M(G)=\lfloor \frac{diam_M(G)+1}{2}\rfloor$.
\qed\end{proof}

We will also need the following lemma from \cite{DraganPhD}. 

\begin{lemma}\cite{DraganPhD}\label{centr-isometr} For every Helly graph $G=(V,E)$ and every set $M\subseteq V$, the graph induced by the center $C_M(G)$ is Helly and it is an isometric (and hence connected) subgraph of $G$.
\end{lemma}

Given this lemma, it will be convenient to denote by $C_M(G)$ not only the set of central vertices but also the subgraph of $G$ induced by this set. Then, $diam(C_M(G))$ denotes the diameter of this graph ($diam(C_M(G))=diam_{C_M(G)}(G)$ by this isometricity). 
\bigskip

Let $\delta(G)$ be the smallest half-integer $\delta\ge 0$ such that $G$ is $\delta$-hyperbolic. Let $\gamma(G)$ be the largest integer $\gamma\ge 0$ such that 
$G$ has a $(\gamma \times \gamma)$ rectilinear grid as an isometric subgraph. Let $\beta(G)$ be the smallest integer $\beta\ge 0$ such that all balls in $G$ are $\beta$-pseudoconvex. Finally, let $\kappa(G)$ be the smallest integer $\kappa\ge 0$ such that $diam(C_M(G))\le \kappa$ for every set $M\subseteq V$.

\begin{theorem} \label{th:parameters} For every Helly graph $G$, a constant bound on one parameter from $\{\delta(G), \gamma(G), \beta(G), \kappa(G)\}$ implies a constant bound on all others.
\end{theorem}

\begin{proof}Let $\delta:=\delta(G), \gamma:=\gamma(G), \beta:=\beta(G), \kappa:=\kappa(G)$. We will show that the following inequalities are true using a few claims:  
$$\kappa\le \min\{2\delta+1, 2\gamma+3, \max\{3, 2\beta+1\}\},$$
$$\beta\le \min\{\max\{0, 2\delta-1\}, 2\gamma+1, \kappa+1\},$$
$$\gamma\le \min\{\delta, \beta, \kappa/2\},$$
$$\delta\le \min\{\gamma, \beta, \kappa/2\}+1.$$

\begin{myclaim} \label{cl:diam-center}
If $G$ is $\delta$-hyperbolic Helly, then $diam(C_M(G)) \leq 2\delta + 1$ and $rad(C_M(G)) \leq \delta + 1$ for every set $M\subseteq V$. In particular, $\kappa\le 2\delta+1$. 
\end{myclaim} 

\noindent
Let $M$ be an arbitrary subset of $V$. Let also $d_G(x,y) = diam_M(G)$ for $x,y \in M$ and let $d_G(u,v) = diam(C_M(G))$ for $u,v \in C_M(G)$.
As $G$ is Helly, $d_G(x,y) \geq 2rad_M(G) - 1$.
We consider the following distance sums:
$d_G(x,y) + d_G(u,v) \geq 2rad_M(G) - 1 + diam(C_M(G))$,
$d_G(x,u) + d_G(v,y) \leq 2rad_M(G)$, and
$d_G(x,v) + d_G(y,u) \leq 2rad_M(G)$.
If $d_G(x,y) + d_G(u,v)$ is not the largest of the three sums, then $diam(C_M(G)) \leq 1$.
Otherwise, by the four point condition, $diam(C_M(G)) \leq 2\delta + 1$.
Consider the pairwise intersecting balls $N^{rad_M(G)}_G[v]$ for all $v \in M$ and $N^{\delta+1}_G[u]$ for all $u \in C_M(G)$. By the Helly property, there is a vertex $c \in C_M(G)$ such that $N^{\delta+1}_G[c] \supseteq C_M(G)$.
As $C_M(G)$ is isometric for any Helly graph (see Lemma \ref{centr-isometr}), the diameter and radius of $C_M(G)$ are realized by paths fully contained in $C_M(G)$.

\begin {myclaim}
If $G$ is $\delta$-hyperbolic, then any ball of $G$ is $(2\delta-1)$-pseudoconvex, when $\delta>0$, and is convex, when $0\le \delta\le 1/2$. In particular, 
$\beta\le \max\{0, 2\delta-1\}$. 
\end {myclaim} 

\noindent
Consider a ball $N^r_G[v]$ centered at a vertex $v \in V$ and with radius $r$.
Let $x,y \in N^r_G[v]$ and let $u \in I_G(x,y)$ be a vertex which is not contained in $N^r_G[v]$.
By contradiction, assume that $d_G(u,x) \geq 2\delta$ and $d_G(u,y) \geq 2\delta$.
Since $u \notin N^r_G[v]$,  $d_G(u,v) > r$. 
Consider the following distance sums:
$d_G(x,y) + d_G(u,v)= d_G(x,u)+d_G(u,y)+ d_G(u,v)>d_G(x,u)+d_G(u,y)+r$,
$d_G(x,u) + d_G(v,y) \leq d_G(x,u) +r$, and
$d_G(x,v) + d_G(y,u) \leq r+ d_G(y,u)$.
Clearly, $d_G(x,y) + d_G(u,v)$ is the largest sum. Without loss of generality, assume that $d_G(x,u) + d_G(v,y)$ is the second largest sum. By the four point condition,
$2\delta\ge d_G(x,y) + d_G(u,v)- d_G(x,u) - d_G(v,y)= d_G(u,y) + d_G(u,v) - d_G(v,y) > 2\delta +r -r=2\delta$, which is not possible. 

\begin {myclaim}
If $G$ is a Helly graph whose all balls are $\beta$-pseudoconvex, then for every set $M\subseteq V$, $diam(C_M(G)) \leq 3$, when $\beta=0$, and $diam(C_M(G)) \leq 2\beta+1$, when $\beta>0$. In particular, $\kappa\le \max\{3, 2\beta+1\}$.
\end {myclaim}

\noindent
Let $M$ be an arbitrary subset of $V$. Let also $d_G(x,y) = diam_M(G)$ for $x,y \in M$ and let $d_G(u,v) = diam(C_M(G))$ for $u,v \in C_M(G)$. 
It is known from \cite{DraganPhD} that if balls in a Helly graph $G$ are convex (i.e., $\beta=0$) then $diam(C_M(G)) \leq 3$. Let now $\beta\ge 1$ and assume, by way of contradiction, that $d_G(u,v)\ge 2\beta+2$. As $G$ is Helly, $d_G(x,y) \geq 2rad_M(G) - 1$. Set $r:=rad_M(G)$.
Note that, since we always have $diam(C_M(G)) \leq 2r$ and we further assume $diam(C_M(G)) \geq 2\beta + 2$, $r \geq 2$.
In particular, $2r-1 \geq r+1$.
Consider the following four balls: $N^{r+1}_G[y]$, $N^{d_G(x,y)-r-1}_G[x]$, $N^{\beta+1}_G[u]$ and $N^{d_G(u,v)-\beta-1}_G[v]$. We show that they pairwise intersect. Clearly, $N^{r+1}_G[y]\cap N^{d_G(x,y)-r-1}_G[x]\neq\emptyset$ and $N^{\beta+1}_G[u]\cap N^{d_G(u,v)-\beta-1}_G[v]\neq\emptyset$. Furthermore, since $d_G(u,x)$, $d_G(u,y)$, $d_G(v,x)$ and $d_G(v,y)$ are at most $r$, the ball $N^{r+1}_G[y]$ intersects both $N^{\beta+1}_G[u]$ and $N^{d_G(u,v)-\beta-1}_G[v]$. As $(d_G(x,y)-r-1)+(\beta+1)\ge 2r-1-r-1+\beta+1=r+\beta-1\ge r\ge d_G(u,x)$, the balls $N^{d_G(x,y)-r-1}_G[x]$ and $N^{\beta+1}_G[u]$ also intersect. Similarly, the balls $N^{d_G(x,y)-r-1}_G[x]$ and $N^{d_G(u,v)-\beta-1}_G[v]$ must intersect as $d_G(u,v)-\beta-1\ge \beta+1$. As all four balls pairwise intersect, by the Helly property, there must exist a vertex $c$ in $G$ such that $d_G(u,v)=d_G(u,c)+d_G(c,v)$ (i.e., $c\in I_G(u,v)$) and $d_G(y,c)=r+1$, $d_G(x,c)=d_G(x,y)-r-1$. 
The latter contradicts with the $\beta$-pseudoconvexity of the ball $N^r_G[y]$ as $u,v$ belong to that ball and $c\in I_G(u,v)$ with $\min\{d_G(c,u),d_G(c,v)\}\ge \beta+1$ is not in $N^r_G[y]$.

\begin {myclaim}
For every graph,  $\gamma\le \min\{\delta, \beta, \kappa/2\}$.
\end {myclaim}

\noindent
Consider an isometric  $(\gamma \times \gamma)$ rectilinear grid in $G$ and let $a,b,c,d$ be the corner vertices of that grid listed in counterclockwise order. The hyperbolicity of $G$ is at least the hyperbolicity of the quadruple $a,b,c,d$ which is exactly  $(d_G(a,c)+d_G(b,d))-(d_G(a,b)+d_G(c,d))=\frac{1}{2}((2\gamma+2\gamma)-(\gamma+\gamma))=\gamma$. Hence, $\gamma\le \delta$. For the ball $N^\gamma_G[a]$, we get: $b,d\in N^\gamma_G[a]$, $c\notin N^\gamma_G[a]$, and $c$ is on a shortest path between $b$ and $d$ with $d_G(c,b)=d_G(c,d)=\gamma$. Hence, $\gamma\le \beta$.
Let now $M=\{a,c\}$. Then, $diam_M(G)=2\gamma$ and $rad_M(G)=\gamma$ and both $b$ and $d$ are in $C_M(G)$. As $d_G(b,d)=2\gamma$, $diam(C_M(G))\ge 2\gamma$, giving $\kappa\ge diam(C_M(G))\ge 2\gamma$.

\begin {myclaim}\label{cl:iso-grid}
For every  Helly graph,  $\delta\le \gamma+1$.
\end {myclaim}

\noindent
This claim follows from a result in \cite{DrGu19}. Let the hyperbolicity of $G$ be $\delta$. According to \cite[Lemma 8]{DrGu19}, if $\delta$ is an integer, then $G$ has an isometric subgraph (named $H_1^{\delta}$ in \cite{DrGu19}), which contains an isometric $(\delta\times \delta)$ rectilinear grid, or an isometric subgraph (named $H_3^{\delta-1}$ in \cite{DrGu19}), which contains an isometric $((\delta-1)\times (\delta-1))$ rectilinear grid. Furthermore, if $\delta$ is a half-integer, then $G$ has an isometric subgraph (named $H_2^{\delta-\frac{1}{2}}$ in \cite{DrGu19}), which contains an isometric $((\delta-\frac{1}{2})\times (\delta-\frac{1}{2}))$ rectilinear grid. Thus, in both cases (whether $\delta$ is a half-integer or an integer), $\gamma\ge \delta-1$ holds.

\begin {myclaim}
For every  Helly graph,  $\delta\le \min\{\beta,\kappa/2\}+1$, $\beta\le \min\{\kappa, 2\gamma\}+1$, $\kappa\le 2\gamma+3$. 
\end {myclaim}

\noindent
These remaining inequalities follow from the previous claims. We have 
$$\delta\le \gamma+1\le\beta+1,$$
$$\delta\le \gamma+1\le\kappa/2+1,$$
$$\beta\le \max\{0,2\delta-1\}\le \max\{0,\kappa+1\}=\kappa+1,$$
$$\beta\le \max\{0,2\delta-1\}\le 2\gamma+1,$$
$$\kappa\le 2\delta+1 \le 2\gamma+3.$$
This concludes the proof of the theorem. \qed\end{proof} 

The following corollaries of Theorem \ref{th:parameters} will play an important role in efficient computations of all eccentricities of a Helly graph. 
Corollary \ref{cor:ub-hyp} gives a sublinear bound on the hyperbolicity of a Helly graph. Corollary \ref{cor:ub-center} gives a sublinear bound on the diameter of the center of a Helly graph.



\begin{corollary}\label{cor:ub-hyp}
The hyperbolicity of an $n$-vertex Helly graph $G$ is at most $\sqrt{n} + 1$.
\end{corollary}

\begin{proof}
If $G$ is $\delta$-hyperbolic then, by Claim~\ref{cl:iso-grid}, $G$ contains an isometric rectilinear grid of side-length $\geq \delta-1$. In particular, $(\delta-1)^2 \leq n$.
\qed\end{proof}

\begin{corollary}\label{cor:ub-center}
For any $n$-vertex Helly graph $G$, we have $diam( C(G) ) \leq 2 \sqrt{n} + 3$.
\end{corollary}

\begin{proof}
We apply Corollary~\ref{cor:ub-hyp} and Claim \ref{cl:diam-center} (with $M=V$). Thus, for a $\delta$-hyperbolic Helly graph $G$, $diam( C(G) ) \leq 2\delta+1\le 2 \sqrt{n} + 3$.
\qed\end{proof}

\section{All eccentricities in Helly graphs}
It is known that the radius (see \cite{DrDu2019-HellyStory}) and a central vertex (see \cite{Du2020-CentMed}) of  an $n$-vertex $m$-edge Helly graph can be computed in $\tilde{\mathcal O}(m\sqrt{n})$-time with high probability. In this section, we improve those results by presenting a deterministic ${\mathcal O}(m\sqrt{n})$ time algorithm which computes not only the radius and a central vertex but also all vertex eccentricities in a Helly graph. 

To show this more general result, we heavily make use of our new structural results from Section~\ref{sec:charc}. In particular, the fact that both the hyperbolicity of a Helly graph $G$ and the diameter of its center $C(G)$ are upper bounded by ${\mathcal O}(\sqrt{n})$ will be very handy. The following results from \cite{DrDu2019-HellyStory}, \cite{Du2020-CentMed}  and  \cite{ChepoiDEHV08-SODA2008,EccentricityTerrain,DrHaVi2018-Revisit} will be also very useful. 



\begin{lemma}\cite{DrDu2019-HellyStory}\label{lem:ecc-threshold}
Let $G$ be an $m$-edge Helly graph and $k$ be a natural number. One can compute the set of all vertices of $G$ of eccentricity at most $k$, and their respective eccentricities, in ${\mathcal O}(km)$ time.  
\end{lemma}

\begin{lemma}\cite{Du2020-CentMed}\label{lem:descend}
Let $G$ be an $m$-edge Helly graph and $v$ be an arbitrary vertex. 
There is an ${\mathcal O}(m)$-time algorithm which either certifies that $v$ is a central vertex of $G$ or finds a neighbor $u$ of $v$ such that $e(u)<e(v)$.  
\end{lemma}

\begin{lemma}\cite{ChepoiDEHV08-SODA2008,EccentricityTerrain,DrHaVi2018-Revisit}\label{lem:hyperb-algo}
Let $G$ be an arbitrary $m$-edge graph and $\delta$ be its hyperbolicity. There is an ${\mathcal O}(\delta m)$-time algorithm which finds in $G$ a vertex $c$ with eccentricity at most $rad(G)+2\delta$.  The algorithm does not need to know the value of $\delta$ in order to work correctly. 
\end{lemma}



First, by combining Lemmas~\ref{lem:descend} and~\ref{lem:hyperb-algo}, we show that a central vertex of a Helly graph $G$ can be computed in ${\mathcal O}(\delta m)$ time, where $\delta$ is the hyperbolicity of $G$. 

\begin{lemma}\label{lem:central}
If $G$ is an $m$-edge Helly graph, then one can compute a central vertex and the radius of $G$ in ${\mathcal O}(\delta m)$ time, where $\delta$ is the hyperbolicity of $G$. 
\end{lemma}
\begin{proof} We use Lemma~\ref{lem:hyperb-algo} in order to find, in ${\mathcal O}(\delta m)$
time, a vertex $c$ of $G$ with eccentricity $e(c)\le rad(G)+2\delta$. Then we apply Lemma~\ref{lem:descend} at most $2\delta$ times in order to descend from $c$ to a central vertex $c^*$. It takes ${\mathcal O}(\delta m)$
time. 
\qed\end{proof}

Combining this with Corollary~\ref{cor:ub-hyp}, we get. 

\begin{corollary}\label{cor:a-center-of-HellyGraph}
For any $n$-vertex $m$-edge  Helly graph $G$, a central vertex and the radius of $G$ can be computed in ${\mathcal O}(m\sqrt{n})$ time. \end{corollary}

We are now ready to prove our main result of this section. 
\begin{theorem}\label{th:ecc-helly}
All vertex eccentricities in an $n$-vertex $m$-edge  Helly graph $G$ can be computed in total ${\mathcal O}(m\sqrt{n})$ time.
\end{theorem}
\begin{proof}
Our goal is to compute $e(v)$ for every $v \in V$. 
For that, we first find a central vertex $c$ and compute the radius $rad(G)$ of $G$, which takes ${\mathcal O}(m\sqrt{n})$ time by Corollary~\ref{cor:a-center-of-HellyGraph}.
If $rad(G) \leq 5\sqrt{n} + 6$ (the choice of this number will be clear later), then $diam(G)\le 2rad(G)\le 10\sqrt{n} + 12$ and we are done by Lemma~\ref{lem:ecc-threshold} (applied for $k = 10\sqrt{n} + 12$); it takes in this case total time ${\mathcal O}(m\sqrt{n})$ to compute all eccentricities in $G$. Thus, from now on, we assume $rad(G) > 5\sqrt{n} + 6$. By Theorem~\ref{th:charcter}(3), for every $v\in V$, $e(v) = d(v,C(G)) + rad(G)$ holds. Thus, 
in order to compute all the eccentricities, it is sufficient to compute $C(G)$. For a central vertex $c \in C(G)$ found earlier, let $S = N^{2\sqrt{n}+3}_G[c]$. Note that, by  Corollary~\ref{cor:ub-center}, $C(G) \subseteq S$.

In what follows, let $r = rad(G)$.
Consider the BFS layers $L_i(S) = \{ v \in V: d(v,S) = i \}$.
Note that if $i \leq r - 4\sqrt{n} - 6 \leq r - diam_S(G)$, then all the vertices of $L_i(S)$ are at distance at most $r$ from all the vertices in $S$.
As a result, in order to compute $C(G) \cap S$ ($=C(G)$), it is sufficient to consider the layers $L_i(S)$, for $i > r - 4 \sqrt{n} - 6$.

Set $A = \bigcup\limits_{i > r - 4\sqrt{n} - 6} L_i(S)$.
Since for every $v \notin S$, $d(v,c) = d(v,S) + 2\sqrt{n} + 3 \leq r$, we deduce that there are at most $(r - 2\sqrt{n} - 3) - (r - 4\sqrt{n} -6) = 2\sqrt{n} + 3$ nonempty layers in $A$. 

We will need to consider the ``critical band'' of all the layers $L_i(S)$, for $1 \leq i \leq r - 4\sqrt{n} - 6$ (all the layers between $S$ and $A$).
We claim that there are at least $\sqrt{n}$ layers in this band.
Indeed, under the above assumption, $r > 5\sqrt{n} + 6$.
Then, the number of layers is exactly $e(c) - 2\sqrt{n} - 3 > 3\sqrt{n} + 3$, minus at most $2\sqrt{n}+3$ layers most distant from $c$ (layers in $A$).
Overall, there are at least $\sqrt{n}$ layers in the critical band, as claimed.

Then, one layer in the critical band, call it $L$, contains at most $n/\sqrt{n} = \sqrt{n}$ vertices.

\begin {myclaim}\label{cl:dist-gate}
For every $a \in A$, there exists a ``distant gate'' $a^* \in Pr(a,L)$ with the following property: $N^r[a] \cap S = N^{r-d(a,L)}[a^*] \cap S$.
\end {myclaim}

In order to prove the claim, set $p = d(a,L) \ \text{and} \ q = d(a,c) \leq r$. 
Let us consider a family of balls $\mathcal{F} = \{N^p[a], \ N^{q-p}[c]\} \cup \{ N^{r-p}[s] \ : \ s \in N^r[a] \cap (S \setminus c) \}$. We stress that $N^p[a] \cap N^{q-p}[c] = Pr(a,L)$. Then, in order to prove the existence of a distant gate, it suffices to prove that the balls in $\mathcal{F}$ intersect; indeed, if it is the case then we may choose for $a^*$ any vertex in the common intersection of the balls in $\mathcal{F}$. 
Clearly, $N^p[a] \cap N^{q-p}[c] \neq \emptyset$ and, in the same way, $N^p[a] \cap N^{r-p}[s] \neq \emptyset$ for each $s \in N^r[a] \cap (S \setminus c)$.
Furthermore, since $L$ is in the critical band, $d(c,L) > 2\sqrt{n} + 3$, and therefore we have for each $s,s' \in S$:
$$2(r-p) \geq 2(q-p) = 2d(c,L)  > diam_{S}(G)\ge d(s,s').$$ In the same way $(q-p) + (r-p) \geq 2(q-p) > diam_{S}(G)\ge d(s,c)$. The latter proves that the balls in $\mathcal{F}$ intersect. This concludes the proof of Claim~\ref{cl:dist-gate}.

We finally explain how to compute these distant gates, and how to use this information in order to compute $S \cap C(G)$.
Specifically:

\begin{itemize}
\item We make a BFS from every $u \in L$. it takes ${\mathcal O}(m|L|) = {\mathcal O}(m\sqrt{n})$ time. Doing so, we can compute $\forall a \in A, \ Pr(a,L)$, in total ${\mathcal O}(|A||L|) = {\mathcal O}(n\sqrt{n})$ time.
\item Since $A$ contains at most ${\mathcal O}(\sqrt{n})$ nonempty layers, then the number of pairwise distinct distances $d(a,L)$, for $a \in A$, is also in ${\mathcal O}(\sqrt{n})$. Call the set of all these distances $I_A$. Then, $\forall u \in L$, and $\forall i \in I_A$, we also compute $p(u,i) = | N_G^{r-i}[u] \cap S |$. For that, we consider the vertices $u \in L$ sequentially. Recall that we computed a BFS tree rooted at $u$. In particular, we can order the vertices of $S$ by increasing distance to $u$. It takes ${\mathcal O}(n)$ time. Similarly, we can order $I_A$ in ${\mathcal O}(\sqrt{n}\log{n}) = o(n)$ time. In order to compute all the values $p(u,i)$, it suffices to scan in parallel these two ordered lists. The running time is $\mathcal{O}(n)$ for every fixed $u \in L$, and so the total running time is $\mathcal{O}(n|L|) = \mathcal{O}(n\sqrt{n})$. 
\item Now, in order to compute a distant gate $a^*$, for $a \in A$, we proceed as follows. Let $i = d(a,L)$. We scan $Pr(a,L)$ and we store a vertex $a^*$ maximizing $p(a^*,i)$. It takes ${\mathcal O}(|A||L|) = {\mathcal O}(n\sqrt{n})$ time. On the way, $\forall u \in L$, let $q(u)$ be the maximum $i$ such that $a^*\equiv u$ is the distant gate of some vertex $a \in A$, such that $d(a,L) = i$ (possibly, $q(u) = 0$ if $u$ was not chosen as the distant gate of any vertex).  
\item Let $s \in S$ be arbitrary. For having $s \in S \cap C(G)$, it is necessary and sufficient to have $s \in N^r[a] \cap S, \forall a \in A$. Equivalently, $\forall u \in L$, one must have $d(s,u) \leq r - q(u)$. This can be checked in time ${\mathcal O}(|L|)$ per vertex in $S$, and so, in total ${\mathcal O}(n\sqrt{n})$ time.
\end{itemize} \qed\end{proof}

\section{Eccentricities in Helly graphs with small hyperbolicity}

In the previous section we showed that a central vertex of a Helly graph $G$ can be computed in ${\mathcal O}(\delta m)$ time, where $\delta$ is the hyperbolicity of $G$. 
This nice result, combined with the property that all Helly graphs have hyperbolicity ${\mathcal O}(\sqrt{n})$ (Corollary~\ref{cor:ub-hyp}), was key to the design of our ${\mathcal O}(m\sqrt{n})$-time algorithm for computing all vertex eccentricities.
Next, we deepen the connection between hyperbolicity and fast eccentricity computation within Helly graphs. 

As we have mentioned earlier,  many graph classes (e.g., interval graphs, chordal graphs, dually chordal graphs, AT-free graphs, weakly chordal graphs and many others) have constant hyperbolicity.  
In particular, the dually chordal graphs and the  $C_4$-free Helly graphs (superclasses to the interval graphs and to the strongly chordal graphs) are proper subclasses of the $1$-hyperbolic Helly graphs. 
This raises the question whether all vertex eccentricities can be computed in linear time in a Helly graph $G$ if its hyperbolicity $\delta$ is a constant.


\medskip

We prove in what follows that it is indeed the case, which is the main result of this section. The following result could also be considered as a parameterized algorithm on Helly graphs with $\delta$ as the  parameter.

\begin{theorem}\label{th:hyp-ecc} If $G$ is an $m$-edge Helly graph of hyperbolicity $\delta$, then the eccentricity of all vertices of $G$ can be computed in ${\mathcal O}(\delta^2 m \log\delta)$ time. 
The algorithm does not need to know the value of $\delta$ in order to work correctly.
If $\delta$ (or a constant approximation of it) is known, then the running time is $\mathcal O(\delta^2 m)$.
\end{theorem}

As a byproduct, we get a linear time algorithm for computing all vertex eccentricities in $C_4$-free Helly graphs as well as in dually chordal graphs, generalizing known results from~ \cite{BrandstadtCD98-dam,Dragan1993-CSJofM,DrDu2019-HellyStory}.  We recall that for dually chordal graphs, until this paper it was only known that a central vertex of such a graph can be found in linear time \cite{BrandstadtCD98-dam,Dragan1993-CSJofM}.

\medskip

The remainder of this section is devoted to proving Theorem~\ref{th:hyp-ecc}. For that, the following result is proved in Subsection~\ref{sec:extract-center}:

\begin{lemma}\label{lem:extract-center}
Let $G$ be an $m$-edge Helly graph, $c$ be a central vertex of $G$ and $k$ be a natural number.
There is an ${\mathcal O}(k^2 m)$-time algorithm which computes $C(G)\cap N^{k}[c]$.
\end{lemma}

\medskip

\begin{proof}[Proof of Theorem~\ref{th:hyp-ecc} assuming Lemma~\ref{lem:extract-center}.]
Since, by Theorem~\ref{th:charcter}(3),  $e(v) = d(v,C(G)) + rad(G)$ holds for every $v\in V$, as before, in order to compute all the eccentricities, it is sufficient to compute $C(G)$. 
We first find a central vertex $c$ and compute the radius $rad(G)$ of $G$. This takes ${\mathcal O}(\delta m)$ time by Lemma \ref{lem:central}.

By Claim \ref{cl:diam-center}, we know that $diam(C(G)) \leq 2\delta + 1$. Therefore, $C(G)\subseteq N^{2\delta+1}[c]$.
\begin{itemize}
    \item If $\delta$ is kown to us, we will fix $k:=2\delta+1$ (if only a constant approximation $\delta'\ge \delta$ of $\delta$ is known, we set $k=2\delta'+1$). Then, we are done applying Lemma~\ref{lem:extract-center}.
    \item Otherwise, we work sequentially with $k=2,3,4,5,8,9,\dots,2^{p},2^{p}+1,2^{p+1},2^{p+1}+1$,\dots, and we stop after finding the smallest integer (power of 2) $k$ such that $C(G)\cap N^{k}[c]=C(G)\cap N^{k+1}[c]$. Indeed, by the isometricity (and hence connectedness) of $C(G)$ in $G$ (see Lemma   \ref{centr-isometr}), the set $C(G)\cap N^{k}[c]$ will contain all central vertices of $G$, i.e., $C(G)\cap N^{k}[c]= C(G)$. The latter will happen for some $k< 2(2\delta+1)$ after at most ${\mathcal O}(\log\delta)$ probes. 
\end{itemize}
Overall, since we need to apply Lemma~\ref{lem:extract-center} at most ${\mathcal O}(\log\delta)$ times, for some values $k< 2(2\delta+1)$, the total running time is ${\mathcal O}(\delta^2 m \log\delta)$. If $\delta$ (or a constant approximation of it is known), then we call Lemma~\ref{lem:extract-center} only once, and therefore the running time goes down to ${\mathcal O}(\delta^2m)$.
\qed\end{proof}

\subsection{Proof of Lemma~\ref{lem:extract-center}}\label{sec:extract-center}

In what follows, $G$ is a Helly graph, $k$ is an integer and $r = rad(G)$. Let $S_k=N^{k}[c]$.
If $r \leq 2k$, we can compute all central vertices in ${\mathcal O}(km)$ time (see Lemma~\ref{lem:ecc-threshold}). Thus from now on, $r > 2k$. 
As $diam_{S_k}(G) \le 2k$, to find all central vertices in $S_k$ (i.e., the set $C(G)\cap S_k$), we will need to consider only the vertices at distance $> r-2k$ from $S_k$. 

Let $i < 2k$ be fixed (we need to consider all possible $i$ between $k$ and $2k-1$ sequentially). Let $A_{k,i} = L_{r-i}(S_k)$ (where we recall that $L_{r-i}(S_k) = \{ v \in V: d(v,S_k) = r-i\}$). We want to compute $$S_{k,i} := \{ s \in S_k : A_{k,i} \subseteq N^r[s] \}.$$ 
Indeed, $C(G)\cap S_k=\bigcap_{i=k}^{2k-1}S_{k,i}$. 


\medskip

The computation of $S_{k,i}$ (for $k,i$ fixed) works by phases. 
We describe below the two main phases of the process.

\medskip
\noindent
\underline{First phase of the algorithm.}
To give the intuition of our approach, we will need the following simple claim.  For a vertex $v\in V$ and an integer $j$, let $L(v,j,S_k):=\{u\in V: d(v,S_k)=d(v,u)+d(u,S_k) \mbox{ and } d(v,u)=j\}$.

\begin {myclaim}\label{cl:smallEccA}
Let $B \subseteq A_{k,i}$ be such that $\bigcap \{ L(b,j,S_k) : b \in B \} \neq \emptyset$, for some $0 \le j < r - i$. Then, for every $s \in S_k$, $\max_{b \in B} d(s,b) \le r$ if and only if $d(s, \bigcap \{ L(b,j,S_k) : b \in B\}) \le r - j$.
\end {myclaim}
\begin{proof}
If $d(s, \bigcap \{ L(b,j,S_k) : b \in B\}) \le r - j$, then clearly $\max_{b \in B} d(s,b) \le r$. Conversely, let us assume $\max_{b \in B} d(s,b) \le r$. 
Set ${\mathcal F} = \{N_G^{r-j}[s], N_G^{r+k-(i+j)}[c]\} \cup \{ N_G^j[b] \ : \ b \in B \}$.
We prove that the balls in ${\mathcal F}$ intersect.
\begin{itemize}
    \item For each $b,b' \in B$, $N^j_G[b] \cap N^j_G[b'] \supseteq \bigcap \{ L(b,j,S_k) : b \in B \} \neq \emptyset$.
    \item Since we assume $\max_{b \in B} d(s,b) \le r$, $N_G^j[b] \cap N_G^{r-j}[s] \neq \emptyset$.
    \item Furthermore, as for each $b \in B$ we have $d(b,c) = d(b,S_k) + k = r-i + k$, we obtain $N_G^{r-i + k - j}[c] \cap N_G^j[b] = L(b,j,S_k) \neq \emptyset$.
    \item Finally, since we have $j < r-i$, $(r-i+k - j) + (r-j) >k+i\geq k \geq d(s,c)$. Therefore, $N_G^{r+k-(i+j)}[c] \cap N_G^{r-j}[s] \neq \emptyset$.
\end{itemize}
It follows from the above that the balls in ${\mathcal F}$ pairwise intersect.
By the Helly property, there exists a vertex $y$ in the common intersection of all the balls in ${\mathcal F}$.
As for each $b \in B$, $y \in N_G^{r-i + k - j}[c] \cap N_G^j[b] = L(b,j,S_k)$, we deduce that $y \in \bigcap \{ L(b,j,S_k) : b \in B \}$. 
Finally, we have $d(s, \bigcap \{ L(b,j,S_k) : b \in B\}) \le d(s,y) \le r-j$.
\qed\end{proof}

We are now ready to present the first phase of our algorithm (for $k,i$ fixed). 
It is divided into $r-i$ steps: from $j = 0$ to $j = r-i-1$.
At step $j$, for $ 0 \le j < r-i$, the intermediate output is a collection of disjoint subsets $V_j^1, V_j^2, ..., V_j^{p_j}$ of the layer $L_{r-i-j}(S_k)$.
These disjoint subsets are in one-to-one correspondence with some partition $B_1, B_2, ..., B_{p_j}$ of $A_{k,i}$.
Specifically, the algorithm ensures that:  
$$\forall 1 \le t \le p_j, \ V_j^t = \bigcap \{ L(b,j,S_k) : b \in B_t \} \neq \emptyset.$$ 

\medskip
Doing so, by the above Claim~\ref{cl:smallEccA}, for any $s \in S_k$ we have $$\max_{z \in A_{k,i}} d(s,z) \le r \iff \max_{1 \le t \le p_j} d(s,V_j^t) \le r - j.$$ 

\medskip
Initially, for $j = 0$, every set $B_t$ is a singleton. Furthermore, $B_t = V_0^t$.
Then, we show how to partition $L_{r-i-(j+1)}(S_k)$ from $V_j^1, V_j^2, ..., V_j^{p_j}$ in total $\mathcal O(\sum_{x \in L_{r-i-j}(S_k)} |N_G(x)|)$ time. Note that in doing so we get a total running time in $\mathcal O(m)$ for that phase.

For that, let us define $W_j^t = N(V_j^t) \cap L_{r-i-(j+1)}(S_r)$. Since the subsets $V_j^t$ are pairwise disjoint, the construction of the $W_j^t$’s takes total $\mathcal O(\sum_{x \in L_{r-i-j}(S_k)} |N_G(x)|)$ time.
Furthermore:

\begin {myclaim}\label{cl:w-subsets}
$W_j^t = \bigcap \{ L(b,j+1,S_k) : b \in B_t \}$.
\end {myclaim}
\begin{proof}
We only need to prove that we have $\bigcap \{ L(b,j+1,S_k) : b \in B_t \} \subseteq W_j^t$ (the other inclusion being trivial by construction). For that, let $x \in \bigcap \{ L(b,j+1,S_k) : b \in B_t \}$ be arbitrary. 
Recall that we have, for each $b \in B_t$, $d(b,c) = k + d(b,S_k) = r-i+k$. In particular, $x \in L(b,j+1,S_k) = L(b,j+1,c)$.
It implies that the balls in $\{N_G[x]$, $N_G^{r-i+k-j}[c]\} \cup \{N_G^j[b] \ : \ b \in B_t\}$ pairwise intersect.
By the Helly property, $x$ has a neighbour in $N_G^{r-i+k-j}[c] \cap \left( \bigcap\{ N^j[b] \ : \ b \in B_t \}  \right) = \bigcap \{ L(b,j,S_k) : b \in B_t \} = V_j^t$.
Since $x \in L_{r-i-(j+1)}(S_k)$, we obtain as desired $x \in W_j^t$.
\qed\end{proof} 

Finally, in order to compute the new sets $V_{j+1}^{t'}$, we proceed as follows. Let $\mathcal W = \{W_j^t : 1 \le t \le p_j\}$. While $\mathcal W \neq \emptyset$, we select some vertex $x \in L_{r-i-(j+1)}(S_k)$ maximizing $\#\{t : x \in W_j^t \}$. Then, we create a new set $\bigcap_{t : x \in W_j^t} W_j^t$, and we remove $\{W_j^t : x \in W_j^t\}$ from $\mathcal W$.
Note that, by the above Claim~\ref{cl:w-subsets}, $\bigcap_{t : x \in W_j^t} W_j^t = \bigcap_{t : x \in W_j^t}\bigcap\{ L(b,j+1,S_k) : b \in B_t\} = \bigcap\{  L(b,j+1,S_k) : b \in \bigcup_{t : x \in W_j^t} B_t\}$.
Furthermore, by maximality of vertex $x$, $\bigcap_{t : x \in W_j^t} W_j^t$ is disjoint from the subsets in $\{W_j^t : x \notin W_j^t\}$.
The latter ensures that all the new sets we create are pairwise disjoint.

In order to implement this above process efficiently, we store each $x \in L_{r-i-(j+1)}(S_k)$ in a list indexed by $\#\{t : x \in W_j^t \}$.
Then, we traverse these lists by decreasing index.
We keep, for each $x \in L_{r-i-(j+1)}(S_k)$, a pointer to its current position in order to dynamically change its list throughout the process.
See also the proof of Lemma~\ref{lem:ecc-threshold} in~\cite{DrDu2019-HellyStory}.
The running time is proportional to $\sum\{ |W_j^t| : 1 \le t \le p_j \} = \mathcal O(\sum_{x \in L_{r-i-j}(S_k)} |N_G(x)|)$.

\medskip
\noindent
\underline{Second phase of the algorithm.}
Let $C_1, C_2, ..., C_p$ denote the sets $V_{r-i-1}^1, ..., V_{r-i-1}^{p_{r-i-1}}$ (i.e., those obtained at the end of  the first phase of our algorithm). Note that $C_1, C_2, ..., C_p$ are subsets of $L_{1}(S_k)$ ($=N_G(S_k)$). At this point, it is not possible anymore to follow the shortest-paths between $A_{k,i}$ and $S_k$. 

Then, let $X = A_{k,i} \cup \{ c \}$. Set $\alpha(c) = k+i+2$ and $\alpha(a) = r$ for each $a \in A_{k,i}$. 
We define the set $Y = \{y : \forall x \in X, d(y,x) \le \alpha(x) \}$.
Observe that $S_{k,i} = Y \cap S_k$ (recall that $S_{k,i}$ was defined as $\{ s \in S_k : A_{k,i} \subseteq N^r[s] \}$).
Therefore, in order to compute $S_{k,i}$, it suffices to compute $Y$.

For that, we proceed in $i+2$ steps. At step $\ell$, for $0 \leq \ell \leq i+1$, we maintain a family of \underline{nonempty pairwise disjoint} sets $Z_\ell^1, Z_\ell^2,...,Z_\ell^{q_\ell}$ and a covering $X_\ell^1, X_\ell^2, \ldots, X_\ell^{q_\ell}$ of $X$ such that the following is true for every $1 \le t \le q_\ell$: 
$$Z_\ell^t = \bigcap_{x \in X_\ell^t} N_G^{\alpha(x) - (i+1) +\ell}[x].$$ 
Doing so, after $i+2$ steps, the set $Y$ is nonempty if and only if $q_{i+1} = 1$ (the above partition is reduced to one group). Furthermore, if it is the case, $Y=Z_{i+1}^1$. 

\medskip
Initially, for $\ell=0$, we start from $Z_0^1 = C_1, ..., Z_0^p=C_p$, and then the corresponding covering is $\forall 1 \le t \le p$, $X_1^t = B_t \cup \{c\}$ (with $B_1, B_2, ..., B_p$ being the partition of $A_{k,i}$ after the first phase of our algorithm). – Note that this is only a covering, and not a partition, because the vertex $c$ is contained in all the groups. –
For going from $\ell$ to $\ell+1$, we proceed as we did during the first phase.
Specifically, for every $t$, let $U_\ell^t = N_G[Z_\ell^t]$.
Since the sets $Z_\ell^t$ are pairwise disjoint, the computation of all the intermediate sets $U_\ell^t$ takes total ${\mathcal O}(m)$ time.

\begin {myclaim}
$U_\ell^t = \bigcap_{x \in X_\ell^t} N_G^{\alpha(x)-(i+1)+(\ell+1)}[x]$.
\end {myclaim}

The proof is similar to that of Claim~\ref{cl:w-subsets}.

\smallskip
Finally, in order to compute the new sets $Z_{\ell+1}^{t'}$, let $\mathcal U = \{U_\ell^t : 1 \le t \le q_\ell\}$. While $\mathcal U \neq \emptyset$, we select some vertex $u \in V$ maximizing $\#\{t : u \in U_\ell^t \}$. Then, we create a new set $\bigcap_{t : u \in U_\ell^t} U_\ell^t$, and we remove $\{U_\ell^t : u \in U_\ell^t\}$ from $\mathcal U$.
The running time is proportional to $\sum\{ |U_\ell^t| : 1 \le t \le q_\ell \} = \mathcal O(m)$.

\medskip
\noindent
\underline{Complexity analysis.} Overall, the first phase runs in ${\mathcal O}(m)$ time, and the second phase runs in ${\mathcal O}(im) = {\mathcal O}(km)$ time. 
Since it applies for  $k,i$ fixed, the total running time of the algorithm of Lemma~\ref{lem:extract-center} (for $k$ fixed) is in $\mathcal O(k^2m)$. 

This completes the proof of Lemma~\ref{lem:extract-center}. 



\bibliographystyle{abbrv}
\bibliography{all-eccentricities-helly}

\end{document}